\documentclass[11pt,letterpaper]{article}
\usepackage{amsthm,color,latexsym,graphicx,url}
\usepackage[margin=1in]{geometry}
\usepackage[font=small,labelfont=bf]{caption}
\usepackage[labelformat=simple]{subcaption}

\usepackage{CJK}

\usepackage{libertine}\usepackage[libertine]{newtxmath}
\usepackage[scaled=0.96]{zi4}
\usepackage[utf8]{inputenc}
\usepackage{fix-cm}

\usepackage{placeins}

\usepackage{lineno}

\title{PSPACE-Hard 2D Super Mario Games: Thirteen Doors}
\author{%
  MIT Hardness Group%
    \thanks{Artificial first author to highlight that the other authors (in
      alphabetical order) worked as an equal group. Please include all
      authors (including this one) in your bibliography, and refer to the
      authors as “MIT Hardness Group” (without “et al.”).}
\and
  Erik D. Demaine%
    \thanks{MIT Computer Science and Artificial Intelligence Laboratory,
      32 Vassar St., Cambridge, MA 02139, USA, \protect\url{{edemaine,hhall314,hayastl}@mit.edu}}
\and
  Lilly Hall\footnotemark[2]
\andlinebreak
  Hayashi Layers\footnotemark[2]
\and
  Matias Korman%
    \thanks{Siemens Electronic Design Automation, Wilsonville, OR 97070, USA,
      \protect\url{matias.korman@siemens.com}}
}
\date{}

\newif\ifabstract
\abstracttrue
\newif\iffull
\ifabstract \fullfalse \else \fulltrue \fi

\usepackage{hyperref}
\hypersetup{breaklinks,bookmarksnumbered,bookmarksopen,bookmarksopenlevel=2}
{\makeatletter \hypersetup{pdftitle={\@title}}}

{\makeatletter
 \gdef\xxxmark{%
   \expandafter\ifx\csname @mpargs\endcsname\relax 
     \expandafter\ifx\csname @captype\endcsname\relax 
       \marginpar{xxx}
     \else
       xxx 
     \fi
   \else
     xxx 
   \fi}
 \gdef\xxx{\@ifnextchar[\xxx@lab\xxx@nolab}
 \long\gdef\xxx@lab[#1]#2{\textbf{[\xxxmark #2 ---{\sc #1}]}}
 \long\gdef\xxx@nolab#1{\textbf{[\xxxmark #1]}}
}

{\makeatletter \gdef\fps@figure{!htbp}}

\let\realbfseries=\bfseries
\def\bfseries{\realbfseries\boldmath}

\def\andlinebreak{\end{tabular}\linebreak\begin{tabular}[t]{c}}

\newtheorem{theorem}{Theorem}[section]

\newtheorem{corollary}[theorem]{Corollary}
\theoremstyle{definition}

\let\epsilon=\varepsilon
\def\defn#1{\textbf{\textit{\boldmath #1}}}
\def\GAME#1{\textit{#1}}
\def\GAME#1{{#1}}
\def\ICON#1{\hspace{0.1em}\raisebox{-0.6ex}{\includegraphics[height=\baselineskip]{figures/icons/#1}}\hspace{0.1em}\nolinebreak}

\usepackage{colortbl}
\definecolor{header}{rgb}{0.29,0,0.51}
\definecolor{gray}{rgb}{0.85,0.85,0.85}
\definecolor{hard}{rgb}{1,0.85,0.85}
\definecolor{open}{rgb}{1,1,0.85}
\definecolor{easy}{rgb}{0.85,0.85,1}
\def\header#1{\cellcolor{header}\textcolor{white}{\textbf{#1}}}

\begin{document}
\maketitle

\begin{abstract}
  We prove PSPACE-hardness for fifteen games in the
  Super Mario Bros.\ 2D platforming video game series.
  Previously, only the original Super Mario Bros.\ was known to be PSPACE-hard
  (FUN 2016), though several of the games we study were known to be NP-hard
  (FUN 2014).
  Our reductions build door gadgets with open, close, and traverse traversals,
  in each case using mechanics unique to the game.
  While some of our door constructions are similar to those from FUN 2016,
  those for Super Mario Bros.\ 2, Super Mario Land 2, Super Mario World 2,
  and the New Super Mario Bros.\ series are quite different;
  notably, the Super Mario Bros.\ 2 door is extremely difficult.
  Doors remain elusive for just two 2D Mario games (Super Mario Land
  and Super Mario Run); we prove that these games are at least NP-hard.
\end{abstract}

\section{Introduction}

At FUN 2016, Demaine, Viglietta, and Williams~\cite{smb1} proved that
it is PSPACE-hard to complete a level in \GAME{Super Mario Bros.},
when the game is generalized to an arbitrary level size, screen size,
number of on-screen enemies, and (exponentially large) time limit.
(They also considered versions with bounded screen size, where off-screen
enemies reset, but this makes the game substantially easier.)
But \GAME{Super Mario Bros.}\ is just the first game in a venerable series of
Super Mario platforming video games.
Consequently, that paper ended with an open problem about other games:

\begin{quote}
Finally, we suspect that our proofs can be adapted to the many
Super Mario Bros.\ sequels, but this remains to be explored.
\hspace{\fill}
\cite{smb1}
\end{quote}

In this paper, we explore these sequels, analyzing the complexity of
all fifteen 2D Super Mario platform video games released to date.
Table~\ref{results} summarizes our results,
most of which are PSPACE-hardness.
Previously, no other PSPACE-hardness results were known,
though four of the games we prove PSPACE-hard were known to be NP-hard
from an earlier FUN 2014 paper \cite{nphard}.

\begin{table}[t]
    \centering
    \begin{tabular}{cl|rcc}
         \header{Year} & \header{Game} & \header{Lower Bound} & \header{Ref} & \header{Previous Bound} \\ \hline
         1985 & Super Mario Bros. & PSPACE-hard & \cite{smb1} & NP-hard \cite{nphard} \\
         1986 & Super Mario Bros.: The Lost Levels & PSPACE-hard & Thm.~\ref{thm:smb2j} & NP-hard \cite{nphard} \\
         1988 & Super Mario Bros.\ 2 & PSPACE-hard & Thm.~\ref{thm:smb2} & NP-hard \cite{nphard} \\
         1988 & Super Mario Bros.\ 3 & PSPACE-hard & Thm.~\ref{thm:smb3} & NP-hard \cite{nphard} \\
         1989 & Super Mario Land & NP-hard & Thm.~\ref{thm:sml} & \\
         1990 & Super Mario World & PSPACE-hard & Thm.~\ref{thm:smw} & NP-hard \cite{nphard} \\
         1992 & Super Mario Land 2: 6 Golden Coins & PSPACE-hard & Thm.~\ref{thm:sml2} & \\
         1995 & Super Mario World 2: Yoshi's Island & PSPACE-hard & Thm.~\ref{thm:smw2} & \\
         2006 & New Super Mario Bros. & PSPACE-hard & Thm.~\ref{thm:nsmb} & \\
         2009 & New Super Mario Bros.\ Wii & PSPACE-hard & Thm.~\ref{thm:nsmb} & \\
         2012 & New Super Mario Bros.\ 2 & PSPACE-hard & Thm.~\ref{thm:nsmb} & \\
         2012 & New Super Mario Bros.\ U & PSPACE-hard & Thm.~\ref{thm:nsmb} & \\
         2015 & Super Mario Maker \em (all four styles) & PSPACE-hard & Thm.~\ref{thm:smm} & \\
         2016 & Super Mario Run & NP-hard & Thm.~\ref{thm:smr} & \\
         2019 & Super Mario Maker 2 \em (all five styles) & PSPACE-hard & Thm.~\ref{thm:smm2} & \\
         2023 & Super Mario Bros.\ Wonder & PSPACE-hard & Thm.~\ref{thm:smbw} & \\
    \end{tabular}
    \caption{New and known results for all sixteen 2D Mario platform games,
      in order of release date.}
    \label{results}
\end{table}

Our PSPACE-hardness reductions all involve building ``door gadgets'',
a technique first used to prove PSPACE-completeness of
Lemmings \cite{Viglietta-2015} and then \GAME{Super Mario Bros.}\
\cite{smb1}.
An \defn{open-close door gadget} is a constant-size piece of a level
that can be in two states, open or closed, and has three possible traversal
paths: the \defn{open} path allows the player to change the state to open,
the \defn{close} path forces the player to change the state to closed,
and the \defn{traverse} path can be traversed only when the door is open.

These original applications also required a ``crossover gadget''
to enable non-interacting crossing tunnels for the player to traverse.
At FUN 2020, however, Ani et al.~\cite{doors} showed that (in most cases)
just a door gadget suffices, and crossovers are unnecessary.
They also introduced two other types of doors --- self-closing doors and
symmetric self-closing doors --- each of which alone suffices to prove
PSPACE-completeness.
They also applied this doors framework to prove that all 3D Mario games
released to date are PSPACE-hard.%
\footnote{Since the paper appeared, one more 3D Mario game has been released:
  \GAME{Bowser's Fury} (as part of \GAME{Super Mario 3D World + Bowser's Fury}).
  But this game has the same mechanics as \GAME{Super Mario 3D World},
  in particular switchboards, so their PSPACE-hardness proof applies.}

In this paper, we apply the doors framework of \cite{doors}
to prove PSPACE-hardness of thirteen more 2D Mario games.
Several of these doors (presented in Section~\ref{sec:open-close})
are variations of the open-close door gadget
from \GAME{Super Mario Bros.}\ \cite{smb1}, but even so,
they require careful adjustments and checking because each game
(except one) adds some mechanics while removing other mechanics
from \GAME{Super Mario Bros.}
In one case, \GAME{Super Mario Bros.\ 2}, the open-close door we construct
(in Section~\ref{sec:smb2}) is completely different and quite complicated.
For other games, we build self-closing doors (in Section~\ref{sec:scd})
or symmetric self-closing doors (in Section~\ref{sec:event sscd}).

For two 2D Mario games, \GAME{Super Mario Land} and \GAME{Super Mario Run},
we have not yet succeeded in building any door gadget.
But we can at least prove only NP-hardness of these games
(in Section~\ref{sec:np}), following the SAT framework first used
to prove \GAME{Super Mario Bros.}\ NP-hard \cite{nphard}.

Our PSPACE-hardness results leave open which Mario games are in PSPACE
and which are harder.  Specifically, membership in PSPACE would hold
if we polynomially bounded the maximum number of on-screen enemies
or the maximum number of enemies at each screen position.
This claim was made for \GAME{Super Mario Bros.}\ in \cite{smb1}.
But even \GAME{Super Mario Bros.}\ has an infinite source of enemies
(if we remove the bound on enemies): Lakitu periodically spawns spinies.
Many other Mario games have pipes that periodically spawn items or enemies.
In some cases, these mechanics can be used to prove
RE-completeness and thus undecidability;
we explore this direction in a companion paper \cite{undecidable}.

\section{Generalized Mario}\label{general}

For each Mario game that we analyze, we make sure to only use blocks, enemies,
objects, and other elements that appear in that game as released.
We also make no changes to the physics or other interactions
between the player and game elements. 

However, actual Mario video games place several constraints on level sizes,
number of onscreen enemies, and other parameters.
For the purposes of analyzing complexity,
we define generalized versions of each game with the following properties:
\begin{itemize}
    \item No arbitrary limits on the level width, level height, and numbers of objects and events.
    \item Exponentially long time limits, or no time limit whatsoever (as in \GAME{Super Mario Bros.\ 2} and \GAME{Super Mario Bros.\ Wonder}).
    \item Arbitrarily large screen size, as large as the entire level.
    One exception is \GAME{Super Mario Bros.\ 2}, which remembers necessary offscreen state, so we do not generalize in this case;
    see Section~\ref{sec:smb2}.
\end{itemize}
For simplicity, we use the original name of each game to refer to the generalized version.
In all of these games, we restrict to a single player using a single input device unless otherwise stated.

\subsection{Forbidding Powerups}
A key defining feature of Mario games is powerups, such as the mushroom which makes Mario grow in size and allows him to take damage once without dying. Powerups can break many of the gadgets presented in this paper. For this reason, we want to assume that Mario comes into each level without a powerup and never collects one, unless we explicitly say otherwise. One way of doing this is to simply build a Mario game which features no powerups. However, in the interest of only caring about solvability of a single level, we should assume that Mario might be able to come into the level with a powerup. We can handle this by starting each level with a powerup and forcing Mario to take damage, e.g., by having to walk over a long row of spikes/munchers or through an enemy.
Henceforth, in all of our gadgets, we will assume Mario begins in the non-powered state.

\subsection{Playtesting}

The majority of the gadgets featured here have been tested in the physics of
the original games, by modifying those games with community-built level editors
(see the Acknowledgments for details).
You can watch videos of the gadgets in action,
under both correct use and attempted misuse, on
\href{https://www.youtube.com/playlist?list=PLCZQ5yzonfsaxrs9jZ41pgMvK4nRHSTXh}{YouTube}.%
\footnote{\url{https://www.youtube.com/playlist?list=PLCZQ5yzonfsaxrs9jZ41pgMvK4nRHSTXh}}
To try these levels out yourself, you can download playable level files from
\href{https://github.com/65440-2023/mario-hardness-gadgets}{GitHub}.%
\footnote{\url{https://github.com/65440-2023/mario-hardness-gadgets}}
Given game hardware constraints and software limits in the released versions
of these games, the gadgets may be modified slightly from what we present
in the paper.

\subsection{Reachability with Doors}
In this paper, we reduce from a known PSPACE-complete problem
called ``reachability with planar door gadgets'' \cite{doors},
which we now describe.

The specific door gadgets used in this paper are the open-close door,
self-closing door, and symmetric self-closing door,
as shown in Figure~\ref{doors}.
Each of these door gadgets has two states: open and closed.
The (optional-open) \defn{open-close door} consists of a traversal path (blue),
a close path (red), and an optional open path (green).
Traversing the close path forces the door into the closed state.
Traversing the open path puts the door into the open state,
but as the entrance and exit location of the open path are the same,
the player can freely decide whether to open the door or not.
The traversal path can be traversed by the player
only when the door is in the open state.
The (optional-open) \defn{self-closing door} consists of a traversal path,
which forces the door into the closed state when traversed,
and an optional open path which opens the door.
A \defn{symmetric self-closing door} consists of two traversal paths.
Traversing the top path is possible only when the door is open, and
it forces the door to become closed; while traversing the bottom path
is possible only when the door is closed, and it forces the door to open.

\begin{figure}
    \centering
    \begin{subfigure}[t]{0.3\linewidth}
        \includegraphics[width=\linewidth]{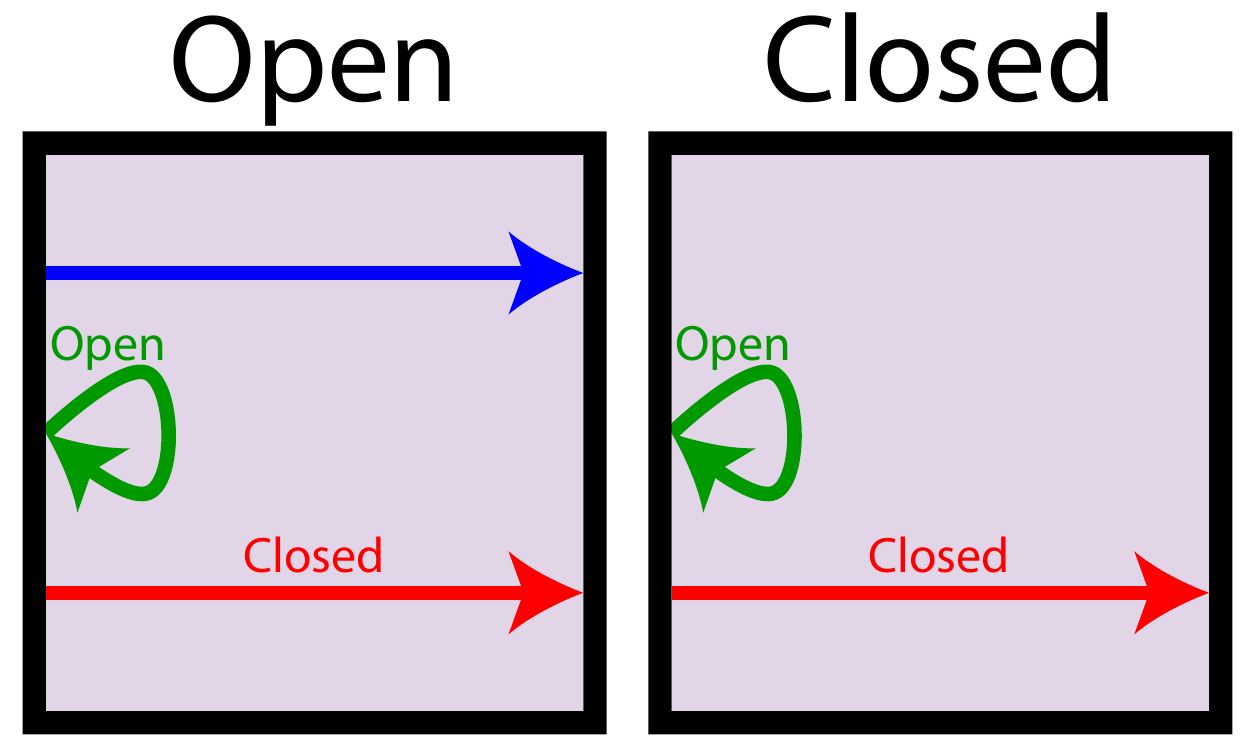}
        \caption{The open-close door gadget}
    \end{subfigure}\hfil
    \begin{subfigure}[t]{0.3\linewidth}
        \includegraphics[width=\linewidth]{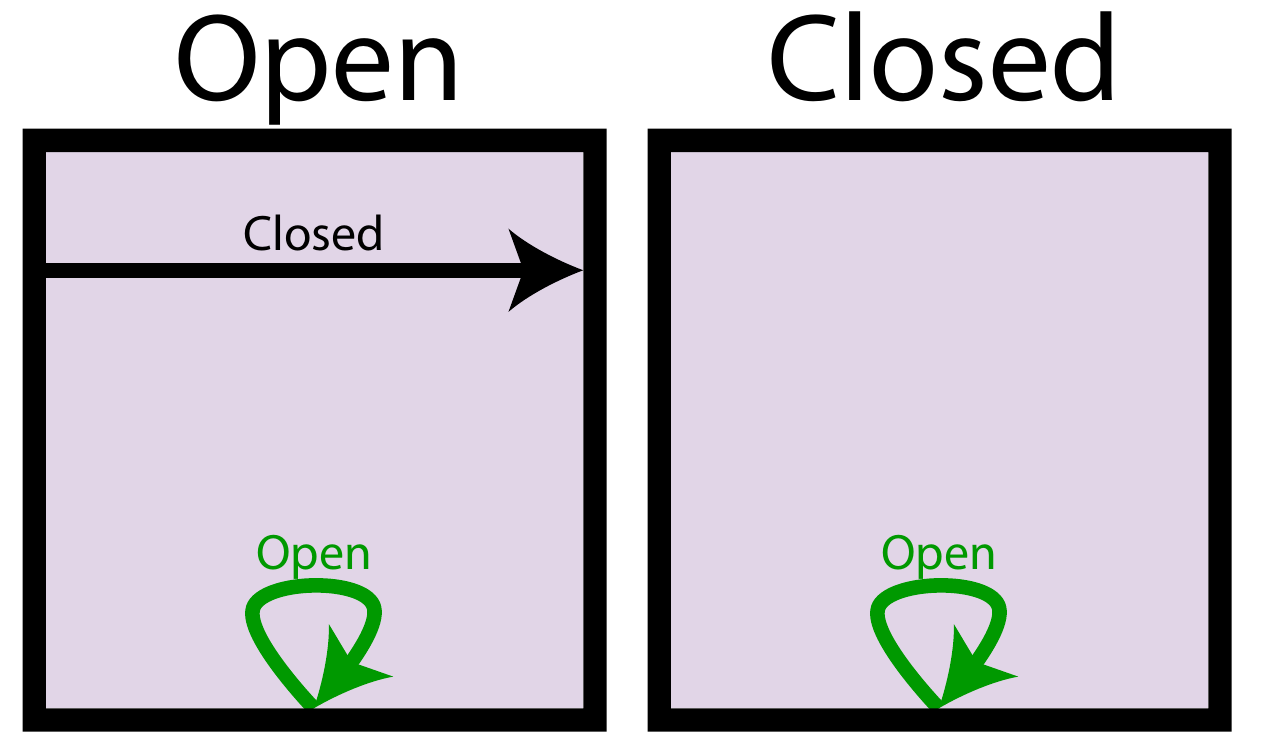}
        \caption{The self-closing door gadget}
    \end{subfigure}\hfil
    \begin{subfigure}[t]{0.3\linewidth}
        \includegraphics[width=\linewidth]{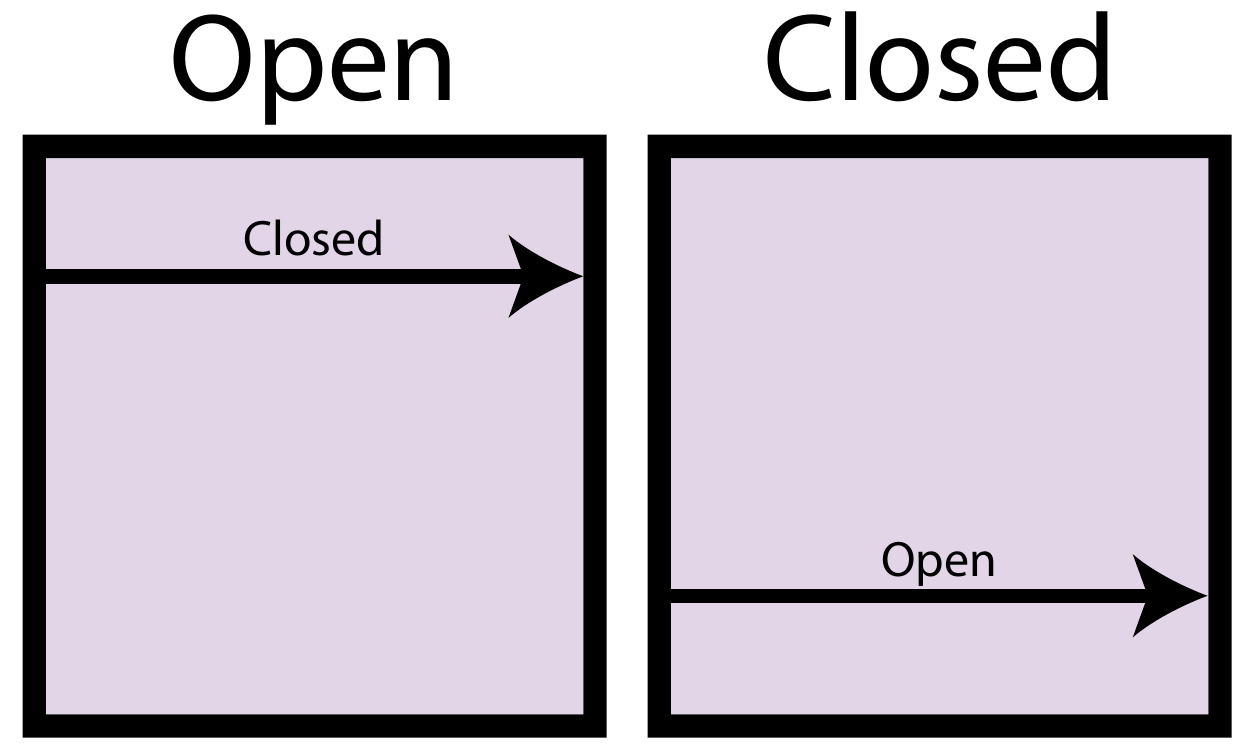}
        \caption{The symmetric self-closing door gadget}
    \end{subfigure}
    \caption{State diagrams for the PSPACE-complete gadgets used in this paper.
      Each box denotes a state (labeled Open or Closed), and each arrow denotes
      a possible transition in that state, labeled with the
      new state that the gadget enters upon such traversal.}
    \label{doors}
\end{figure}

Now we are given a \defn{system} of door gadgets,
consisting of several instances of the same type of door gadget,
an initial state for each gadget,
and a graph defining connections between locations
(entrances and exits) of the doors.
In a \defn{planar} system, these connections do not cross each other or
the door gadgets themselves.
The \defn{reachability} problem asks, given a system and two locations,
whether it is possible for the player to start at the first location
and reach the second location, by a sequence of traversals of door gadgets
and connection edges.
In \defn{planar reachability}, we restrict to planar systems of gadgets.

For all three types of door gadgets described above,
planar reachability is PSPACE-complete \cite{doors}.
(That paper defines one open-close door gadget
which it does not prove PSPACE-complete, but it has a non-optional open path.)
If we can build any door gadget in a Super Mario platform game,
then we can instantiate this gadget multiple times and connect them together
via tunnels made of blocks in such a way
that the only way for Mario to reach the flagpole is
via suitable traversal through these gadgets,
and thereby prove PSPACE-hardness of the Mario game.

\section{Standard Open-Close Doors}
\label{sec:open-close}

We will begin by examining games which can build doors with a similar structure to that featured in \cite{smb1}. All of these doors are open-close doors. These doors have traverse and open on their left side, and close on the right. An enemy is hit from below while on a certain type of block such as a brick block to cause it bounce up and move between the two sides of the door. While on the left, the enemy blocks access to the traverse tunnel, and while on the right, the enemy blocks access to the close tunnel. There is some object which prevents Mario from crossing between the traverse and close tunnels but which allows for the enemy to pass through. In most of these games, the enemy in question becomes stunned for a constant amount of time when hit from below, and if the player is able to hit it and traverse quickly enough to pick it up, this could break the door. For this reason, we assume all tunnels in these gadgets are significantly longer than pictured, such that they take longer to traverse than it does for the enemy to wake up. This is indicated in our gadget figures with ellipses ($\cdots$).

\subsection{Super Mario Bros.: The Lost Levels}

This game was originally released in Japan as
\begin{CJK}{UTF8}{min}スーパーマリオブラザーズ２\end{CJK} (Super Mario Bros.\ 2),
but was renamed to \GAME{Super Mario Bros.: The Lost Levels}
when it was finally released in North America,
as part of the 1993 compilation \GAME{Super Mario All-Stars}.
We follow the latter naming convention, to avoid confusion
with the other game named (uniquely) \GAME{Super Mario Bros.\ 2}
(covered in Section~\ref{sec:smb2}).

\begin{theorem} \label{thm:smb2j}
Super Mario Bros.: The Lost Levels is PSPACE-hard.
\end{theorem}
\begin{proof}
\GAME{Super Mario Bros.: The Lost Levels} is almost exactly the same as \GAME{Super Mario Bros.} but with different levels, graphics, and only a few minor tweaks. In particular, the physics of Mario's movement is the same, and spinies and firebars behave in the same way. Hence, we can build the same door used in \cite{smb1} to show that \GAME{Super Mario Bros.: The Lost Levels} is hard. For reference, Figure~\ref{smb2j} shows such a gadget built in \GAME{Super Mario Bros.: The Lost Levels}.
The idea is that the door is open whenever the spiny \ICON{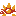}
is on the right side of the central fire \ICON{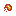} (as in the figure),
in which case Mario can freely follow the traverse path.
To traverse the close path, Mario must jump to hit the brick block
\ICON{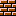} with careful timing so that the spiny
above gets bumped to the other side of the gadget,
closing the gadget and enabling Mario to reach the close exit.
Similarly, visiting the open path allows Mario to hit the
spiny to the other side, opening the gadget.
\end{proof}

\begin{figure}[ht]
    \centering
    \includegraphics[width=\linewidth]{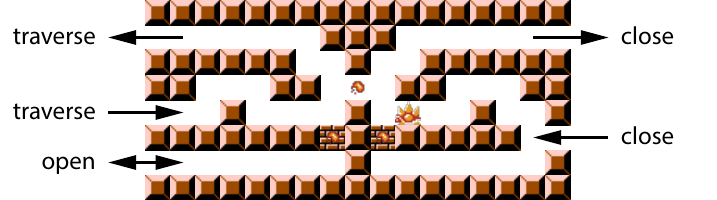}
    \caption{An open-close door in \GAME{Super Mario Bros.: The Lost Levels}}
    \label{smb2j}
\end{figure}

\subsection{Super Mario Bros.\ 3}

\begin{theorem} \label{thm:smb3}
Super Mario Bros.\ 3 is PSPACE-hard.
\end{theorem}

\begin{proof}
Figure~\ref{smb3} shows our door gadget for \GAME{Super Mario Bros.\ 3}. Notably, we make use of \emph{clouds} \ICON{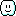}, which are a semi-solid platform. Objects (including both Mario and enemies) can pass through clouds from below and the sides, but not from above. When Mario hits a brick block \ICON{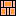} with a spiny on top of it, the spiny is bounced on top of the cloud. After a brief period of waking up, the spiny walks to the opposite side of the gadget. The black plants are munchers \ICON{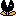}, an enemy which will damage Mario if he touches them. Mario cannot pass above the munchers in the center of the gadget without touching them and dying, but the spiny is invulnerable to their effects and can pass through without issue. We also make use of invisible coin blocks \ICON{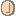}. When Mario hits these from below, they become solid blocks, but before that point, they are completely intangible, and the spiny can pass through them from the side with no issues. Their purpose is to restrict Mario from jumping over the spiny \ICON{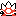} and onto the cloud. Because of them, Mario can only pass through the traverse or close tunnel by jumping next to the munchers, and this can only happen if the spiny is on the opposite side of the gadget.

\begin{figure}[ht]
    \centering
    \includegraphics[width=\linewidth]{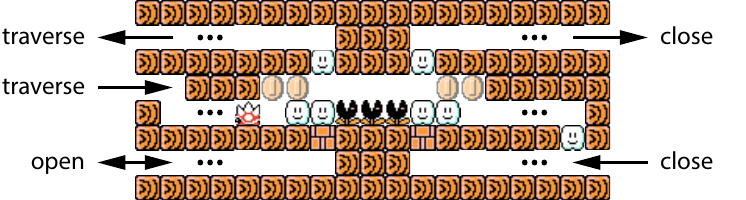}
    \caption{An open-close door in \GAME{Super Mario Bros.\ 3}}
    \label{smb3}
\end{figure}

One might worry that Mario can hit the coin blocks on traversal, permanently breaking the gadget. This is true. If Mario hits a coin block, the gadget will enter a broken state where it is stuck open or closed, depending on the location of the spiny. However, this is not an issue, as doing so as it can only make the reachability problem more restrictive. If Mario hits a coin block while traversing on the left side of the gadget, the spiny will never fall down the left side, and the door will be permanently open. This means that if Mario attempts to close the door, he will become stuck and forced to backtrack or die. A similar argument shows that it is never in Mario's interest to force the door into a stuck closed state.
\end{proof}

\subsection{New Super Mario Bros.\ Series}

\begin{corollary} \label{cor:nsmb}
\GAME{New Super Mario Bros.}, \GAME{New Super Mario Bros.\ Wii},
\GAME{New Super Mario Bros.\ 2}, and \GAME{New Super Mario Bros.\ U}
are all PSPACE-hard.
\end{corollary}

\begin{proof}
The door from Theorem~\ref{smb3}
works in all of the \GAME{New Super Mario Bros.}\ games for the same reasons.
In place of clouds, we use a different semisolid platform \ICON{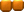}, and some of the \GAME{New Super Mario Bros.}\ games feature spikes \ICON{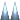} instead of munchers \ICON{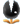}, but the objects behind these graphical changes function in the same way. For reference, Figure~\ref{nsmbw} shows a door built in \GAME{New Super Mario Bros.\ Wii}.
\end{proof}

\begin{figure}[ht]
    \centering
    \includegraphics[width=\linewidth]{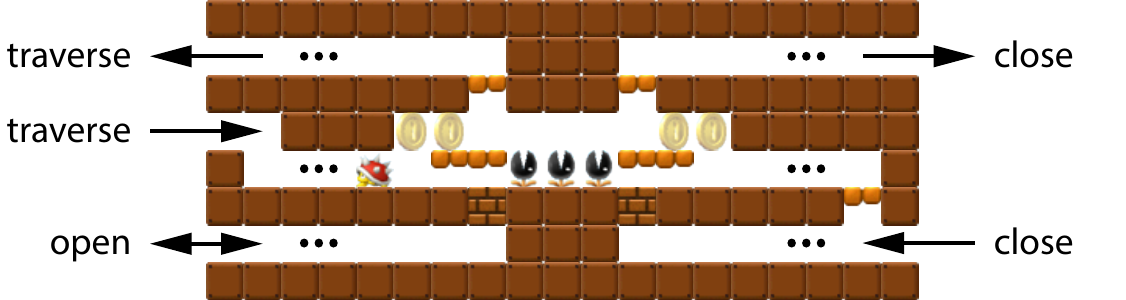}
    \caption{An open-close door in \GAME{New Super Mario Bros.\ Wii}}
    \label{nsmbw}
\end{figure}

We will give another proof of this result in Section~\ref{sec:event sscd},
which uses different mechanics and is furthermore robust against the mechanics of Wii U gamepads.

\subsection{Super Mario World}

\begin{theorem} \label{thm:smw}
Super Mario World is PSPACE-hard.
\end{theorem}

\begin{proof}
In \GAME{Super Mario World}, we cannot use spinies because they die instantly upon being hit from below. Instead, we make use of the goomba \ICON{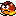} enemy, which behaves differently in \GAME{Super Mario World} in that it becomes stunned instead of killed when hit. Figure~\ref{smw} shows the resulting door gadget. Unlike with spinies, Mario can safely jump on goombas. To ensure Mario does not jump over the goomba or step on it to stun it and pass through a blocked tunnel, we place munchers \ICON{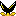} at the top of the traversal tunnels. If Mario would jump on the goomba, it will bounce him into the munchers, killing him instantly. The only valid way to traverse the gadget is to first open it by hitting the on/off block \ICON{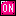} from below when a goomba is on it. Doing so will cause the goomba to bounce onto the munchers in the center of the gadget and walk across to the other side, allowing safe traversal. The close traversal works analogously. As a remark, the on/off switches contain extra functionality which we completely ignore here; we only care about them in terms of their ability to bounce goombas when hit from below.
\begin{figure}[ht]
    \centering
    \includegraphics[width=\linewidth]{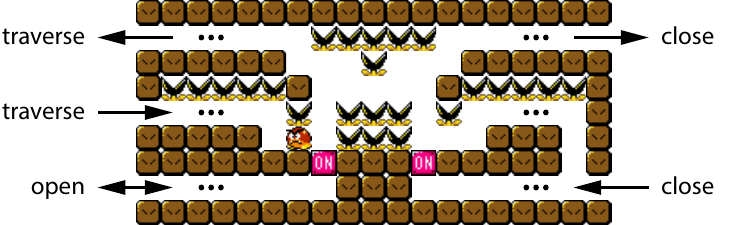}
    \caption{An open-close door in \GAME{Super Mario World}}
    \label{smw}
\end{figure}
\end{proof}

\subsection{Super Mario Maker}

\begin{theorem} \label{thm:smm}
All styles of Super Mario Maker are PSPACE-hard.
\end{theorem}
\begin{proof}
Our door gadget for \GAME{Super Mario Maker} is almost identical to the door gadget for \GAME{Super Mario Bros.\ 3}. The main difference is that we add \emph{one-ways} \raisebox{-0.6ex}{\includegraphics[height=\baselineskip]{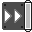}}, which only allow one-directional traversal, immediately below and above the brick blocks \ICON{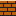}. This is because, in the \GAME{Super Mario World} style, the blocks corresponding to brick blocks temporarily become intangible when hit, allowing a spiny or Mario to pass through. The one-ways ensure that Mario cannot go up and the spinies cannot go down through these blocks, but Mario maintains the ability to bounce the spiny from below. All features present in this gadget exist in all four styles of \GAME{Super Mario Maker}, so this gadget works in all of them. Figure~\ref{smm} shows a construction of the gadget.

\begin{figure}[ht]
    \centering
    \includegraphics[width=\linewidth]{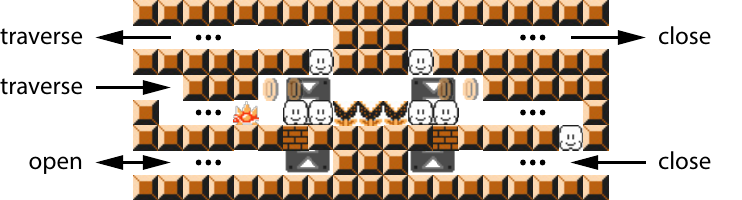}
    \caption{An open-close door in \GAME{Super Mario Maker}}
    \label{smm}
\end{figure}
\end{proof}

\begin{corollary} \label{cor:smm2-2d}
The four 2D styles Super Mario Maker 2 are PSPACE-hard.
\end{corollary}
\begin{proof}
Each of the four 2D styles of \GAME{Super Mario Maker 2} features all elements of the door gadget pictured in Figure~\ref{smm}, behaving in the same way, so we use the same gadget gadget from \GAME{Super Mario Maker}. As a small detail, if we choose, we can simplify the construction involving 2 one-ways and a brick block by using the on/off switches present in \GAME{Super Mario Maker 2} since they cannot be passed through in any game style and bounce spinies in the same way (these are the same type of blocks we used in the reduction to show that \GAME{Super Mario World} is PSPACE-hard).
\end{proof}

\begin{theorem} \label{thm:smm2}
    All styles of Super Mario Maker 2 are PSPACE-hard.
\end{theorem}
\begin{proof}
Corollary~\ref{cor:smm2-2d} proves hardness for all 2D styles. There is one remaining style present in \GAME{Super Mario Maker 2}, the \GAME{Super Mario 3D World} style. This style works very differently from the other four styles, so we construct a different door gadget, shown in Figure~\ref{smm3dw}. The issue in the \GAME{Super Mario 3D World} style is that the spiny \ICON{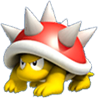} changes direction when it wakes up and is unable to cross the gadget on its own. Fortunately, there are on/off switches \ICON{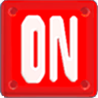} which toggle the state of the blue squares \ICON{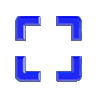} pictured. When the switch is on, as pictured, the blue squares are empty. When the switch is off, they are replaced by solid blue blocks. To cause the spiny to change sides, the player hits the brown brick block \ICON{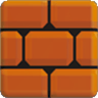}, bounding the spiny into the center of the gadget. The player can then toggle the state of the switched blocks to allow the spiny to walk across to the other side of the gadget. Then, toggling the switch again will drop the spiny into the desired side. Hitting the switch toggles the state of all blue blocks in the level, including those in other gadgets, but these do not interact with the other gadgets because the blue blocks are above the spinies, so only the door where the player hits the spiny is affected.

With the addition of the switched blocks, it is possible for Mario to hit the switch at such a time that the spiny is crushed by the blocks. If this happens, the door can no longer be closed, which is bad. In \GAME{Super Mario Maker}, enemies can be created holding keys \ICON{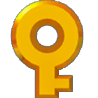} which Mario will collect if the spiny dies. Mario will keep this key and any others he collects until they are used on locked doors or warp blocks \ICON{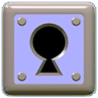}. We give the spiny in our gadget a key, and at the very end of the level we force Mario through a ``check" gadget which consists of a 1 tall path with a locked warp box (see Figure~\ref{smm3dw-check}). If the player has a key, they get warped and become completely stuck. If they have no key, they pass through with no issue. Taking advantage of the check gadget, we use a second key to prevent the player from switching between the two sides of the gadget, taking the place of the munchers in the previous reduction.
\end{proof}

\begin{figure}[ht]
    \centering
    \includegraphics[width=\linewidth]{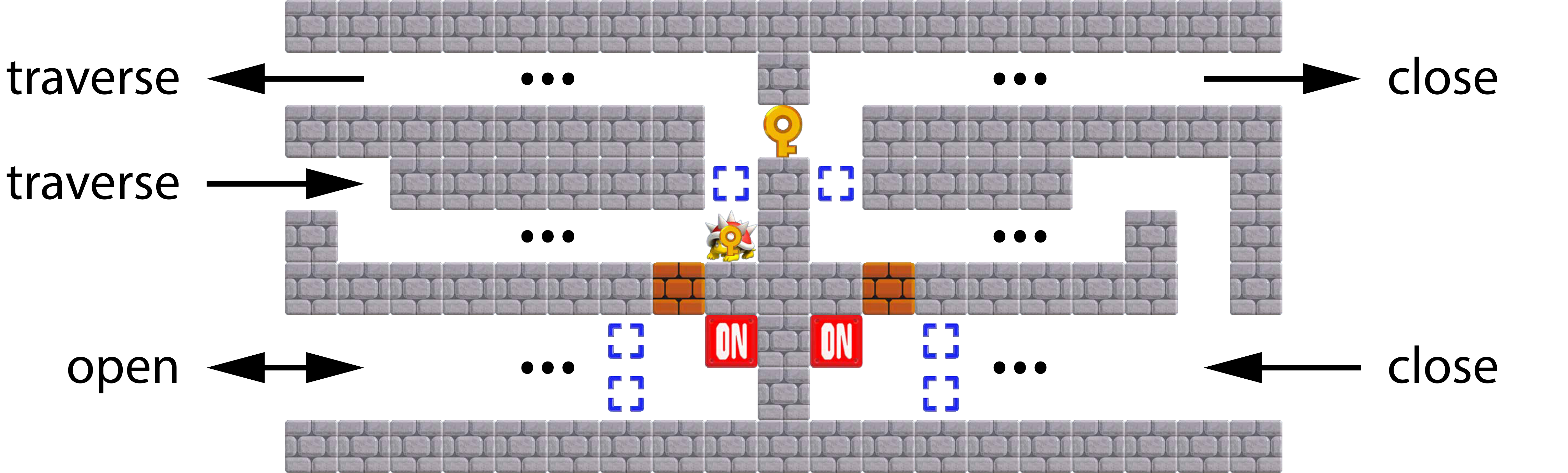}
    \caption{An open-close door in the Super Mario 3D World style of \GAME{Super Mario Maker 2}}
    \label{smm3dw}
\end{figure}
\begin{figure}[ht]
    \centering
    \includegraphics[width=0.5\linewidth]{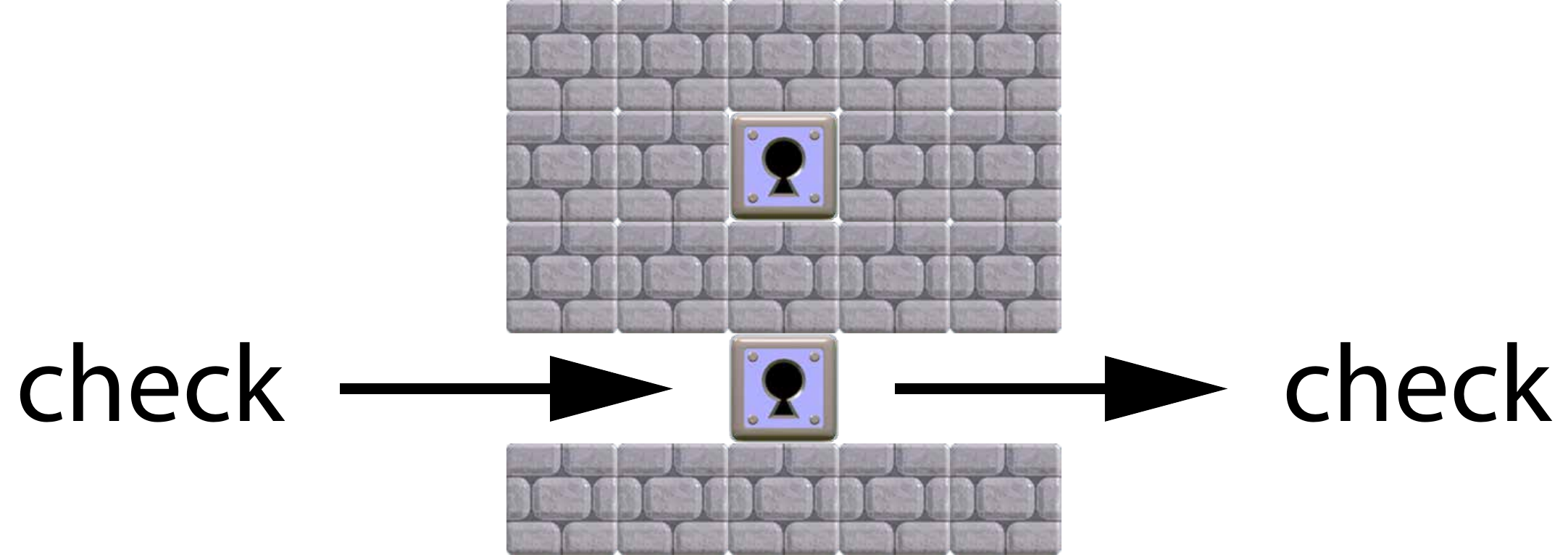}
    \caption{The check gadget for the Super Mario 3D World style of \GAME{Super Mario Maker 2}}
    \label{smm3dw-check}
\end{figure}

\FloatBarrier
\subsection{Super Mario Bros.\ Wonder}

\begin{theorem} \label{thm:smbw}
Super Mario Bros.\ Wonder is PSPACE-hard.
\end{theorem}
\begin{proof}
Our door gadget for \GAME{Super Mario Bros.\ Wonder} is similar to the door gadget for \GAME{Super Mario Bros.\ 3}, but we take advantage of the fact that \GAME{Super Mario Bros.\ Wonder} can have non-grid-aligned objects to simplify the gadget by removing the need for invisible coin blocks. Figure~\ref{smbw} shows a construction of the gadget.
\end{proof}

\begin{figure}[ht]
    \centering
    \includegraphics[width=\linewidth]{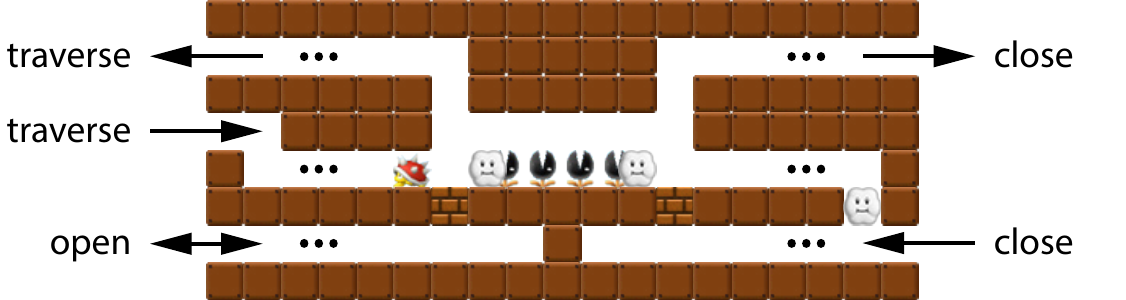}
    \caption{An open-close door in \GAME{Super Mario Bros.\ Wonder}}
    \label{smbw}
\end{figure}

\FloatBarrier

\section{Self-Closing Doors}
\label{sec:scd}

\subsection{Super Mario Land 2}

\begin{theorem} \label{thm:sml2}
    Super Mario Land 2: 6 Golden Coins is PSPACE-hard
\end{theorem}
\begin{proof}
To show \GAME{Super Mario Land 2: 6 Golden Coins} is PSPACE-hard, we reduce from planar reachability with self-closing doors \cite{doors}, using the self-closing door gadget shown in Figure~\ref{sml2}. \GAME{Super Mario Land 2} does not have a spiny enemy, so we instead use a koopa \ICON{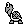}. In this game, koopas do not walk off of ledges (like red koopas in other games), which complicates matters. When the door is closed (the koopas is in the top left, pictured solid in Figure~\ref{sml2}), it blocks the traversal path because Mario is too tall to jump past it in the 1-block corridor. To open the door, Mario enters the open tunnel, and with a well-timed hit, it is possible to bounce the koopa into the right of the gadget (as pictured partially transparent in Figure~\ref{sml2}). This opens up the beginning of the traversal tunnel, but blocks the right half. To continue with traversal, Mario must head down the vertical tunnel, below the semisolid \ICON{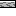}, and reach the other light gray block to bounce the koopa back to the left. Mario can then backtrack and finish traversal. However, when Mario bounces the koopa to the left, it is possible for him to instead time the hit such that the koopa falls down the same tunnel Mario entered. If this happens, the koopa will block Mario's path out, which prevents Mario from finishing the traversal. Therefore, an optimal player must choose not to do this, and we can safely assume it does not happen.
\begin{figure}[ht]
    \centering
    \includegraphics[width=\linewidth]{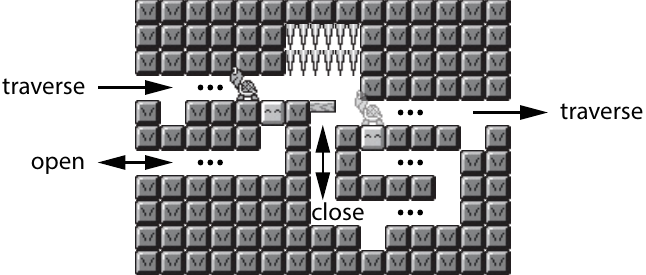}
    \caption{A self-closing door in \GAME{Super Mario Land 2}}
    \label{sml2}
\end{figure}
\end{proof}

\subsection{Super Mario World 2: Yoshi's Island}

\begin{theorem} \label{thm:smw2}
Super Mario World 2: Yoshi's Island is PSPACE-hard
\end{theorem}

\begin{proof}
To show PSPACE-hardness for \GAME{Super Mario World 2: Yoshi's Island}, we reduce from planar reachability with self-closing doors \cite{doors}, using the self-closing door gadget shown in Figure~\ref{smw2}. To traverse, the player must ride the chomp rock \ICON{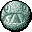}, a spherical enemy which rolls when stood on, to the right side of the gadget. Attempts to cross without using the chomp rock will fail as the player does not have enough height to reach the traversal exit without the chomp rock and cannot stand on spikes \ICON{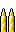} without dying. The chomp rock rolls down a slope and loses momentum upon reaching the bottom due to partially hitting the spikes. The player can then use it to traverse, but is unable to push the chomp rock back up the slope at the same time. To open the gadget, the player gains access to a reusable helicopter powerup \ICON{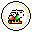}. This powerup briefly turns Yoshi into a helicopter, allowing the player to safely push the the chomp rock up the slope and reset the door while hovering safely above the spikes. When the timer on the helicopter powerup runs out, Yoshi is instantly transported back to the open tunnel. Because of the powerup timer, by making all tunnels between gadgets sufficiently long, we can ensure that the player cannot reach and open any gadgets other than the intended one.
\begin{figure}[ht]
    \centering
    \includegraphics[width=\linewidth]{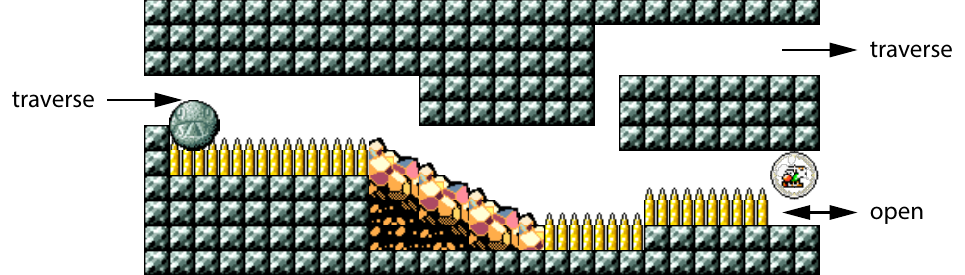}
    \caption{A self-closing door in Super Mario World 2}
    \label{smw2}
\end{figure}
\end{proof}

\section{Event-Based Symmetric Self-Closing Doors}
\label{sec:event sscd}

Corollary~\ref{cor:nsmb} already showed that all four
\GAME{New Super Mario Bros.}\ games are PSPACE-hard.
But these games also implement a system called \defn{events},
which allow the player to interact with blocks via switches and
event controllers.
In this section, we develop event-based symmetric self-closing door gadget,
for two main reasons:
\begin{enumerate}
    \item These doors are incredibly simple and small compared to the previous ones.
    \item Unlike enemies, events persist regardless of the location of the screen. This means that we do not need to generalize screen size, which shows hardness for something that is much closer to the original game. This is also of some importance for \GAME{New Super Mario Bros.\ U}, as discussed in Section~\ref{nsmbugamepad}.
\end{enumerate}

That being said, there are some drawbacks to using events. For one, they work behind the scenes and their behavior is not transparent to the player. It could be seen as not in the spirit of the game to build doors out of objects that are not true visible game elements. Furthermore, while these games do not impose arbitrary limits on the number of enemies (outside of memory constraints), they do impose a limit of 255 events based on the encoding of object data in binary. To allow for arbitrary events, we would have to generalize to a system which allows for more events.
For these reasons, this door does not completely supplant
the door of Corollary~\ref{cor:nsmb}.

\subsection{Event Mechanics}
The relevant event mechanics used in this paper are as follows:
\begin{itemize}
    \item There are numbered events, each of which can be either on or off
    \item There are numbered locations, which are sections of the level bounded by some rectangle
    \item Location controllers toggle an event based on the status of enemies or players in a given location
    \item There are blocks which appear or disappear according to the status of an event
\end{itemize}

\subsection{\GAME{New Super Mario Bros.\ U} with Gamepad} \label{nsmbugamepad}
In \GAME{New Super Mario Bros.\ U}, the Wii U gamepad serves a special purpose. If the player is playing on a non-gamepad controller, they (or a friend) can use the gamepad in ``boost mode" to help Mario. The player with the gamepad can create platforms on-screen to help the player cross dangerous areas, and if Mario steps on 10 of these, the gamepad player gains access to a boost star which allows them to kill enemies by tapping on them. It is very reasonable to disallow boost mode in our generalized \GAME{New Super Mario Bros.\ U} as it somewhat breaks the restriction of only one player, and this must be done for the door in Figure~\ref{nsmbw} to work properly, but with events, we can allow for the gamepad as the gamepad player has no control over any event behavior.

\subsection{Self-Closing Door}

\begin{theorem} \label{thm:nsmb}
    \GAME{New Super Mario Bros.}, \GAME{New Super Mario Bros.\ Wii},
    \GAME{New Super Mario Bros.\ 2}, and \GAME{New Super Mario Bros.\ U}
    are all PSPACE-hard, even when restricted to just blocks and events.
\end{theorem}

\begin{proof}
The gadget pictured in Figure~\ref{nsmbu} is an event-based symmetric self-closing door. It uses one event. The blocks labeled 1 are on if and only if the event is on, and the blocks labeled 2 are on if and only if the event is off. When the player enters location 3, the event turns on, and when the player enters location 4, the event turns off. In this way, traversal through the top path is only possible when the event is off, and doing so turns the event on, and traversal through the bottom path is only possible when the event is on, and doing so turns the event off. The door in Figure~\ref{nsmbu} is pictured in \GAME{New Super Mario Bros.\ U}, but the same gadget works in all four \GAME{New Super Mario Bros.}\ games.
\end{proof}

\begin{figure}[ht]
    \centering
    \includegraphics[width=0.75\linewidth]{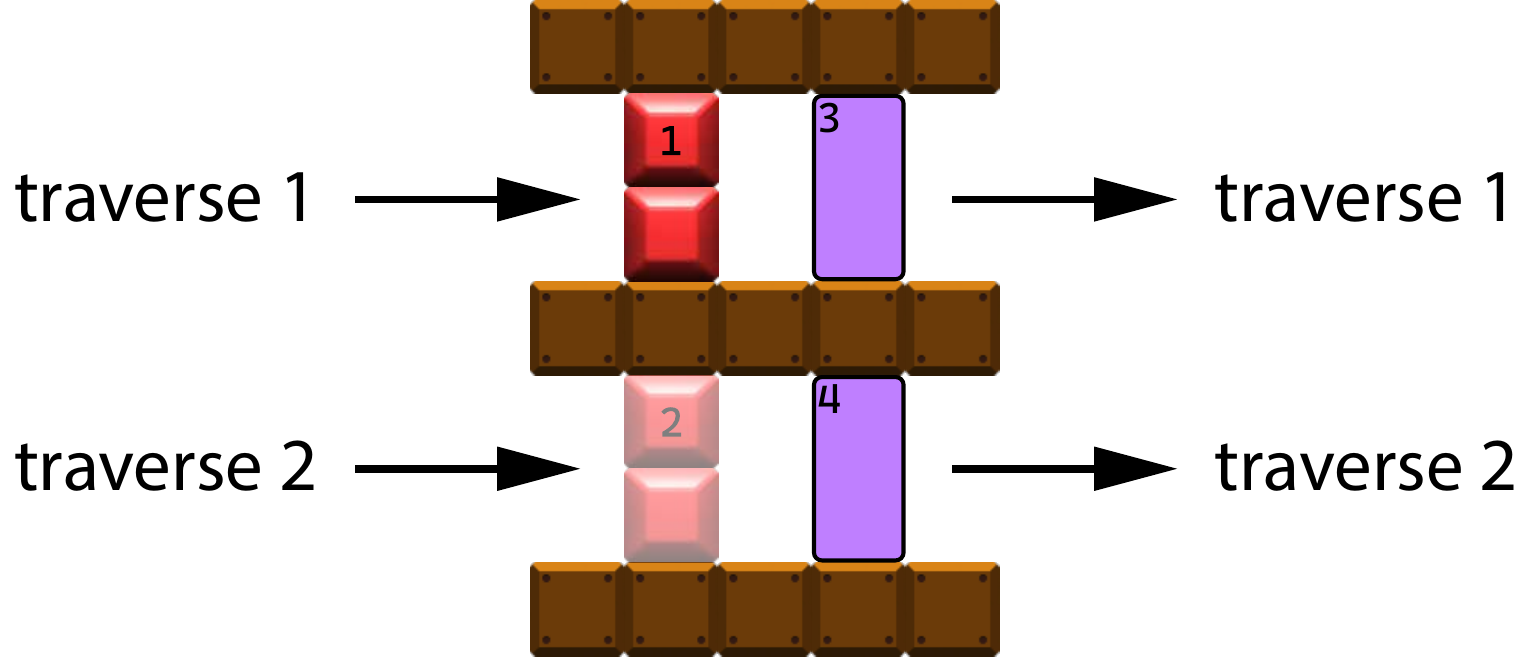}
    \caption{A symmetric self-closing door in \GAME{New Super Mario Bros.\ U}}
    \label{nsmbu}
\end{figure}
\FloatBarrier

\section{Super Mario Bros.\ 2 Open-Close Door}
\label{sec:smb2}

In this section, we prove PSPACE-hardness of \GAME{Super Mario Bros 2}
(as named in North America).  This reduction consists of
our most complicated door gadget by far.

\subsection{Mechanics}
\GAME{Super Mario Bros 2} has significantly different mechanics from the rest of the games in this paper, and our reduction for it is relies on many of these mechanics, so we take some extra time to describe the behavior of various objects. Refer to Figure~\ref{smb2u} to see images of relevant objects.
\begin{itemize}
    \item The blue blocks with an X \ICON{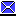} pictured in our gadgets are semi-solids or diodes. They can be passed through from below, but not from above.
    \item The spikes \ICON{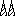} will kill the player if they land on them from above.
    \item Mario cannot naturally pass through a 1 block tall tunnel. However, by crouching, holding in the direction of desired movement, and repeatedly jumping, Mario can make very slow progress through the tunnel. This mechanic is used to pass through the tunnel at the bottom right of the gadget.
    \item The mushrooms \ICON{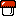} are objects which the player can pick up and place down. When placed, they fall and then snap to the grid in the level. When the player is holding any object, they cannot hold any other object or enemy. They also cannot pass through 1 block tall tunnels.
    \item The Birdo enemy \ICON{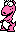}, pictured in bottom right of the gadget will shoot eggs \ICON{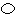} when on screen. Eggs move across the screen horizontally at a constant speed until they hit a solid object. The player can jump off of them and can also pick them up and throw them to kill enemies.
    \item There are two types of shy guy enemies. They walk back and forth and can be picked up and thrown. When thrown, they return to walking. The red shy guy \ICON{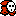} will walk off of ledges (analogous to green koopas in other Mario games), and the pink shy guy \ICON{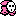} will turn when it reaches a ledge (analogous to red koopas in other Mario games). Mario can also stand on shy guys and ride them across spikes.
    \item By crouching for a brief period of time to charge his jump, Mario can perform a super jump, which is an especially tall jump.
    \item The game features 4 playable characters: Mario, Luigi, Toad, and Peach. The relevant differences for this gadget are that Luigi jumps higher than the others and that Peach can float for a brief period of time. Our gadget is robust to use of all four characters.
    \item The game remembers the position of mushrooms that are moved when they become off screen, so we do not need to generalize screen size. In fact, we will take advantage of screen size constraints.
\end{itemize}

\subsection{Glitches}

Super Mario Bros.\ 2 has many glitches and unintended behavior. Several of these are documented in \cite{glitches}. As far as we are aware, there are two which are relevant for our gadget.
\begin{itemize}
    \item If the player jumps and makes contact with the corner of an enemy, they can perform a second jump midair, resulting in a much higher jump than should be allowed. We avoid abuse of this by not giving the player opportunities to jump off of the corners of enemies where important.
    \item If the player is crouching on top of an enemy which walks into a one block tall tunnel, the player will not collide with the wall, despite doing so visually. This does not happen if the player is holding an item. If the roof of the tunnel is only one block thick, the player can then jump through the ceiling. We make tunnel ceilings two blocks thick and make use of shy guys to prevent abuse of this.
\end{itemize}
\begin{figure}[ht]
    \centering
    \includegraphics[width=1\linewidth]{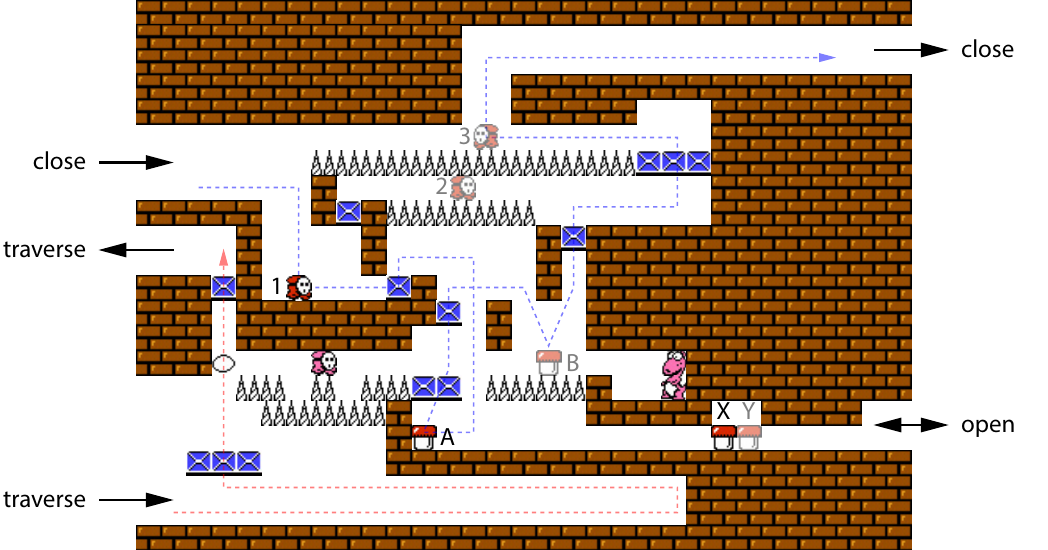}
    \caption{An open-close door in Super Mario Bros.\ 2}
    \label{smb2u}
\end{figure}

\subsection{Open-Close Door}

\begin{theorem} \label{thm:smb2}
Super Mario Bros.\ 2 is PSPACE-hard.
\end{theorem}
\begin{proof}
We reduce from planar reachability with open-close doors \cite{doors}. A door is pictured in Figure~\ref{smb2u}. The state of the gadget is determined by a mushroom object, which the player can pick up and place in different locations. Specifically, the mushroom in position \textsf{A}, representing the open state, can instead be in position \textsf{B}, representing the closed state. When the door is open, the player can run through the bottom tunnel until Birdo is onscreen. Birdo will then shoot an egg to the left which the player can jump on to reach the exit of the traverse tunnel. This path is pictured in red. When the door is closed, the mushroom is in position \textsf{B}, and the egg's path is blocked, preventing traversal. Even Luigi with a super jump is unable to reach the platform without the egg.

Closing the door begins with the player throwing the shy guy in position \textsf{1} onto the spikes above, as shown in position \textsf{2}. The shy guy walks to the right for use later. The player cannot cross this same gap because of the spikes. The glitch which allows for the player to ride an enemy across does not apply because the ceiling prevents Mario from getting on the enemy in the first place. Then, the player proceeds to the right and moves the mushroom from position \textsf{A} to \textsf{B}. Once the mushroom is in position \textsf{B}, they can use it to super jump up to the right. If they do not place the mushroom, they will be unable to exit, as a regular jump off of an egg does not give enough height, even for Luigi, to exit, and it takes too long to charge a super jump for the player to perform one off of the egg before it moves too far to the left. Hence, for the player to continue through the close tunnel, they must place the mushroom in location \textsf{B}, blocking the path of the eggs, and closing the traverse tunnel. Finally, the player can move the shy guy that was thrown to position \textsf{2} to the row of spikes above, and jump off of it in position \textsf{3} to cross an otherwise impassable spike tunnel. The shy guy continues walking to the left and falls back down to position \textsf{1} in preparation for the next use. An outline of this path is pictured in blue. Note that the player can choose not to jump off the shy guy and instead ride it back to the entrance of the close gadget, but this does not accomplish anything productive as they are simply forced to close the gadget again, so we can assume this does not happen.

Opening the door begins on the bottom right of the gadget with a 1-toggle sub-gadget. The mushroom can be placed in either position \textsf{X} or \textsf{Y}, but can only be picked up from the top, not the side. The effect of this is that the 1-toggle can only be traversed from the right if the mushroom is in position \textsf{X} and from the left if in position \textsf{Y}. This prevents the player from exiting through the open tunnel when closing the door. For the player to traverse from right to left in the open path, the mushroom is forced into position \textsf{Y}. Then, the player can (optionally) move the other mushroom from position \textsf{B} to position \textsf{A}, opening the door. They cannot exit through the close tunnel as they do not have the shy guy required to cross the spikes, and are forced to return through the 1-toggle, returning it to position \textsf{X} as they leave.

One might worry that the player can ride an egg through the gadget to the traverse path when opening or closing the door. This would be an issue, but there is a pink shy guy in the egg's path. The egg will pass through with no issue, but if a player is riding an egg, the taller hitbox of the shy guy will cause them to become stuck in the middle of a row of spikes, and they are unable to jump across the shy guy fast enough to continue riding the egg. The shy guy is placed far enough to the left such that Birdo is off screen and will not shoot any more eggs while the player is stranded. As pictured, the player could grab an egg and throw it at the shy guy before performing this shortcut, thus clearing the way, but making the tunnel sufficiently long and placing the shy guy far enough from the ends will prevent this. We simply show the smaller gadget here for simplicity.
\end{proof}

\FloatBarrier

\section{NP-Hardness}
\label{sec:np}

We initially set out to prove all 2D Mario Games PSPACE-complete. Sadly, we have not yet succeeded for two of these games we considered: \GAME{Super Mario Land} and \GAME{Super Mario Run}. Nonetheless, we can at least prove NP-hardness for both of these games. We will be using the framework developed in \cite{nphard}.

\begin{theorem} \label{thm:sml}
    Super Mario Land is NP-hard.
\end{theorem}

\begin{proof}
    The framework for proving NP-hardness requires us to create the following gadgets:
    \begin{itemize}
        \item \textbf{Start and Finish.}\quad We use the trivial start and end gadgets.
        \item \textbf{Variable.}\quad We use the same variable gadget as used in \cite{nphard}.
        \item \textbf{Clause and Check.}\quad We use the clause gadget pictured in Figure~\ref{smlclause}. Each of the 3 question blocks \ICON{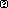} in the left side of the gadget creates a mushroom which can be used to damage boost through the spikes \ICON{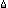} at the right of the gadget. However, if there is no mushroom, the player is unable to cross.

\begin{figure}[ht]
    \centering
    \includegraphics[width=0.75\linewidth]{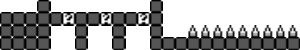}
    \caption{A clause gadget in Super Mario Land}
    \label{smlclause}
\end{figure}

        \item \textbf{Crossover.}\quad We use the crossover gadget pictured in Figure~\ref{smlcross}. This is a unidirectional crossover in which left to right traversal happens before top to bottom. When coming from the left, the player gains access to a star in a question block. This allows them to pass to the right over the row of spikes, which is much longer than pictured such that the player can only barely make it across with the star. They are unable to go up as they do not have a mushroom. To traverse from bottom to top, the player gains access to a mushroom in a question block which allows them to break the breakable blocks \ICON{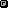} and move up, where they then take forced damage and are returned to a powered-down state. However, they are unable to cross the spikes to the right because the maximum distance Mario can run on spikes with a star is mush longer than the distance he damage boost with a mushroom, and the player is unable to reach the star block on the left side of the gadget. After the player completes a vertical traversal, there is leakage from horizontal to vertical, but, as discussed in \cite{nphard}, this is not an issue.
\end{itemize}

\begin{figure}[ht]
    \centering
    \includegraphics[width=0.75\linewidth]{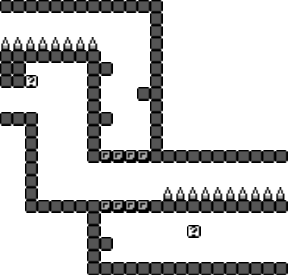}
    \caption{A crossover gadget in Super Mario Land}
    \label{smlcross}
\end{figure}
The combination of these gadgets is sufficient to show NP-hardness.
\end{proof}

\begin{theorem} \label{thm:smr}
    Super Mario Run is NP-hard.
\end{theorem}

\begin{proof}
First, we note that Super Mario Run is unique compared with other Mario games in that the player can, under normal circumstances, only move right (and is pulled to the right by default). Fortunately, the game provides backflip \ICON{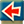} blocks which will cause the player to move left when jumping off of them, which allows for easy creation of wires which allow for moving left, as pictured in Figure~\ref{smrwire}. Vertical wires are traversed via wall jumping or falling.

The framework for proving NP-hardness requires us to create the following gadgets:
    \begin{itemize}
        \item \textbf{Start and Finish.}\quad We use the trivial start and end gadgets.
        \item \textbf{Variable.}\quad We use the variable gadget pictured in Figure~\ref{smrvar}. The forced rightward movement inherently creates a diode, ensuring that once the player chooses either the top or bottom path, they cannot choose the other variable.
        \item \textbf{Clause and Check.}\quad We use the clause gadget pictured in Figure~\ref{smrclause}. This works almost identically to the clause gadget in \cite{nphard}, but with grinders \ICON{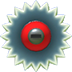} in place of fire bars to enforce that the player must have a star.
        \item \textbf{Crossover.}\quad We use the clause gadget pictured in Figure~\ref{smrcross}. Vertical traversal is accomplished either via falling or wall jumping, and is bidirectional. The traversal from left to right requires hitting a switch \ICON{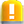} to briefly toggle the states of the red \ICON{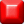} and outlined \ICON{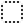} blocks, creating a tunnel from left to right. Because of the forced right condition, the player moves through the tunnel before the switch state can revert, and is unable to stall for leakage.
\end{itemize}

\begin{figure}
    \centering
    
    \begin{subfigure}[t]{0.45\linewidth}
        \centering
        \includegraphics[width=\linewidth]{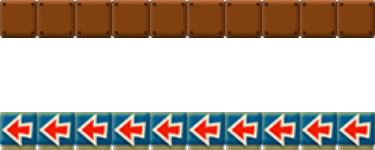}
        \caption{A left-moving wire}
        \label{smrwire}
    \end{subfigure}\hfill
    \begin{subfigure}[t]{0.45\linewidth}
        \centering
        \includegraphics[width=\linewidth]{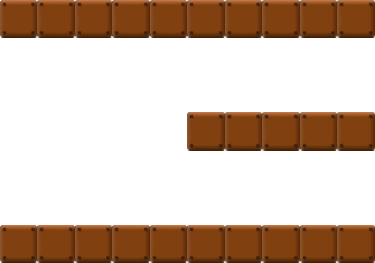}
        \caption{A variable gadget}
        \label{smrvar}
    \end{subfigure}

    \medskip

    \begin{subfigure}[t]{0.6\linewidth}
        \centering
        \includegraphics[width=\linewidth]{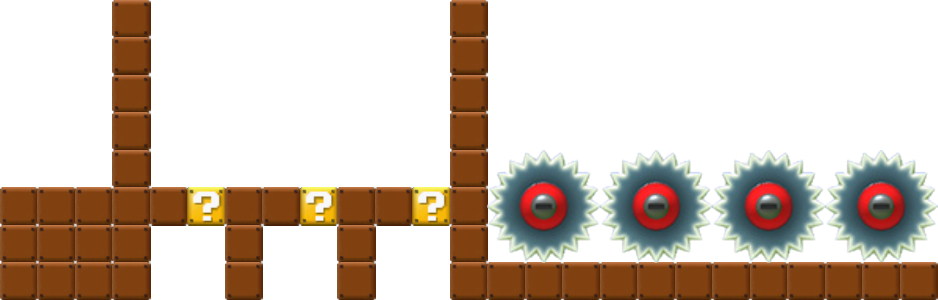}
        \caption{A clause gadget}
        \label{smrclause}
    \end{subfigure}\hfill
    \begin{subfigure}[t]{0.3\linewidth}
        \centering
        \includegraphics[width=\linewidth]{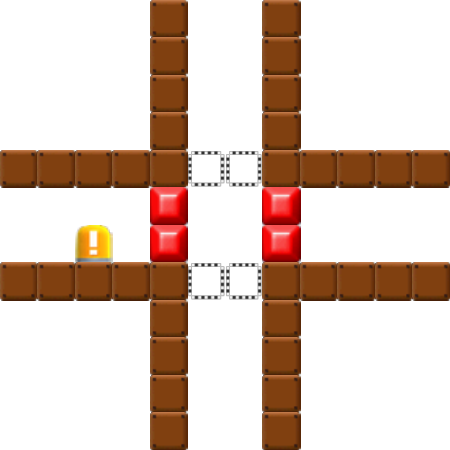}
        \caption{A crossover gadget}
        \label{smrcross}
    \end{subfigure}
    \caption{Super Mario Run Gadgets}
    \label{smr}
\end{figure}

The combination of these gadgets is sufficient to show NP-hardness.
\end{proof}

\FloatBarrier

\section{Open Problems}
\label{sec:open}

As discussed in the previous section, while we initially set out to prove all 2D Mario Games PSPACE-complete, \GAME{Super Mario Land} and \GAME{Super Mario Run} have only been shown to be NP-hard. Super Mario Land features a very different set of objects from the other Mario games, and none that stand out to us as being able to toggle the state of a door an unbounded number of times. \GAME{Super Mario Run} poses a similar issue, with enemies behaving somewhat differently from the other games. It also suffers from the issue of access: unlike all other games considered in this paper, there are no known tools for creating custom \GAME{Super Mario Run} levels.

In this paper, we did not place any restrictions on generating new enemies, although none of our doors take advantage of this. If we do enforce that no new objects can be spawned, then we get containment in NPSPACE as the size of the level cannot grow, and hence PSPACE by Savitch's Theorem. In this case, we have PSPACE-completeness for every PSPACE-hard game examined in this paper. However, most of these games have some mechanism by which the level size could grow without bound, with the removal of arbitrary object limits. For example, the original paper on \GAME{Super Mario Bros.} PSPACE-hardness, \cite{smb1}, makes claims that Super Mario Bros.\ is contained in NPSPACE based on its levels only taking up a polynomial amount of space. However, this is not entirely true as \GAME{Super Mario Bros.} also has the Lakitu enemy which can create additional spinies. With the exponentially long timer, a player could stand near a Lakitu for exponentially long and generate exponentially many spinies which walk off-screen. These will require an exponential amount of memory, since we assume in SMB-General that the game remembers the positions of all off-screen enemies. With an exponentially large timer, all of these games are clearly contained in NEXPTIME, but this is a very arbitrary limit, and in the version without a timer, there are no obvious upper bounds other than RE for any of these games.\footnote{In fact, we explore RE-completeness for some of these games in \cite{undecidable}.}

\section*{Acknowledgments}

This paper was initiated during open problem solving in the MIT class
on Algorithmic Lower Bounds: Fun with Hardness Proofs (6.5440),
taught by Erik Demaine in Fall 2023;
and during the 33rd Bellairs Winter Workshop on
Computational Geometry, co-organized by Erik Demaine and Godfried Toussaint
in March 2018 in Holetown, Barbados.
We thank the other participants of that class and workshop
for helpful discussions and providing an inspiring atmosphere.

Much of the work in this paper would not have been possible without the wealth of tools created by the communities surrounding the Super Mario games. In particular, the following level editing and emulation tools were of immense importance to testing various gadget ideas and mechanics:
\begin{itemize}
    \item \textit{The NES Super Mario Brothers 2 Level Editor} by loginsinex --- \url{https://github.com/loginsinex/smb2}
    \item \textit{SMB3 Foundry} by mchlnix --- \url{https://github.com/mchlnix/SMB3-Foundry}
    \item \textit{MarCas} by Coolman --- \url{https://www.romhacking.net/utilities/518/}
    \item \textit{Lunar Magic} by FuSoYa --- \url{https://fusoya.eludevisibility.org/lm/}
    \item \textit{Golden Egg} by Romi --- \url{https://www.smwcentral.net/?p=section&a=details&id=4645}
    \item \textit{Reggie!} by the NSMBW Community --- \url{https://github.com/NSMBW-Community/Reggie-Updated}
    \item \textit{CoinKiller} by Arisotura --- \url{https://github.com/Arisotura/CoinKiller}
    \item \textit{Miyamoto} by aboood40091 --- \url{https://github.com/aboood40091/Miyamoto}
    \item \textit{BizHawk} by TASEmulators --- \url{https://github.com/TASEmulators/BizHawk}
    \item \textit{Mesenrta} by threecreepio --- \url{https://github.com/threecreepio/mesenrta}
    \item \textit{Mesenrta-s} by threecreepio --- \url{https://github.com/threecreepio/mesenrta-s}
    \item \textit{mGBA} by endrift --- \url{https://mgba.io/}
    \item \textit{Dolphin} by the Dolphin Emulator Project --- \url{https://dolphin-emu.org/}
    \item \textit{Citra} by Citra Team --- \url{https://citra-emu.org/}
    \item \textit{Cemu} by Team Cemu --- \url{https://cemu.info/}
\end{itemize}

\bibliography{citations}
\bibliographystyle{alpha-key}

\end{document}